\documentclass[a4paper,USenglish,cleveref, autoref, thm-restate]{lipics-v2019}

\title{Fork-Resilient Cross-Blockchain Transactions through Algebraic Topology} 

\titlerunning{Fork-Resilient Cross-Blockchain Transactions through Algebraic Topology} 

\author{Dongfang Zhao}
{University of Nevada, United States}
{dzhao@unr.edu}
{https://orcid.org/0000-0002-0677-634X}
{U.S. Department of Energy, contract No. DE-SC0020455}

\authorrunning{D. Zhao} 

\Copyright{Dongfang Zhao} 



\category{} 

\relatedversion{} 

\supplement{}

\nolinenumbers 

\hideLIPIcs  

\EventEditors{John Q. Open and Joan R. Access}
\EventNoEds{2}
\EventLongTitle{42nd Conference on Very Important Topics (CVIT 2016)}
\EventShortTitle{CVIT 2016}
\EventAcronym{CVIT}
\EventYear{2016}
\EventDate{December 24--27, 2016}
\EventLocation{Little Whinging, United Kingdom}
\EventLogo{}
\SeriesVolume{42}
\ArticleNo{23}

\begin{document}

\maketitle

\begin{abstract}
The cross-blockchain transaction (CBT) serves as a cornerstone for the next-generation, blockchain-based data management systems.
However, state-of-the-art CBT models do not address the effect of the possible local fork suspension that might invalidate the entire CBT.
This paper takes an algebraic-topological approach to abstract the blockchains and their transactions into simplicial complexes
and shows that CBTs cannot complete in either a \textit{committed} or an \textit{aborted} status by a $t$-resilient message-passing protocol. 
This result implies that a more sophisticated model is in need to support CBTs and, thus, sheds light on the future blockchain designs.
\end{abstract}

\section{Introduction}


The cross-blockchain transaction~\cite{dzhao_cidr20} serve as a cornerstone for the next-generation, blockchain-based data management systems:
the inter-blockchain operations would enable the interoperability among distinct, potentially heterogeneous, blockchains.
The state-of-the-art blockchain implementation can only support two-party transactions between two distinct blockchains through the \textit{sidechain} protocol~\cite{sidechain},
incurring considerable latency in terms of hours and yet acceptable for the targeting cryptocurrency applications~\cite{cosmos}.
One recent work by Herlihy~\cite{mherlihy_podc18} notably studied how to support general cross-blockchain operations among an arbitrary number of distinct blockchains through serialized hash locks implemented in smart contracts,
assuming a relaxed semantics on the atomicity of operations.
Later, Zakhary et al.~\cite{vzakhary_arxiv19} proposed a 2PC-based protocol to support both parallelism and atomicity.
One denominator of these recent works is that they did not consider the possible \textit{forks} commonly seen in blockchain implementations:
a cross-blockchain transaction can still be invalidated if part of its ``local'' changes on some blockchains is committed but then suspended due to the fork competition \textit{within} a blockchain.
To this end, Zhao~\cite{dzhao_arxiv20_cbt_pst} proposed a point-set-topological approach to map the fork-induced topological space to the transaction's topological space---providing a powerful tool to study a fork-resilient CBT through topological equivalence, i.e., homeomorphism.

This paper takes into account the possible local fork suspension and analyzes the \textit{completeness} of CBTs:
whether a CBT can proceed to a \textit{completed} or \textit{aborted} final status.
Our assumption of the underlying computation model is as follows:
a message-passing communication model, an asynchronous timing model, and a crash-failure adversary model with $t$-resilience, $2t < (n+1)$ where $n \in \mathbb{Z}_+$ and $(n+1)$ is the total number of nodes.
We take an algebraic-topological approach to abstract the blockchains and their transactions into simplicial complexes,
and show that CBTs cannot complete in either a \textit{committed} or an \textit{aborted} status by a $t$-resilient message-passing protocol. 
This result, thus, implies that a more sophisticated model is in need to support CBTs in the face of local fork suspension.

\section{Cross-Blockchain Transactions with Fork Suspension}

\subsection{Models}
\label{subsec:cbt_notation}

We denote the set of distinct blockchains $\cal C$, whose cardinality is at least two:
$|\cal C| \ge$ 2.
Each blockchain is an element $C_i \in \cal C$,
where $0 \leq i \leq n = |\mathcal{C}| - 1$.
Each blockchain is a list of blocks,
each of which is identified by its index $j$:
$C_i = \left( v_i^0, \ldots, v_i^j\right)$.
Of note, $v_i^0$ is also called the \textit{genesis block} of $C_i$ in the literature of blockchains.
An \textit{$(n+1)$-party transaction} carried out on $\cal C$ touches one and only one block at each blockchain.
Specifically, an $(n+1)$-party global transaction $T$ can be represented as a set of $n+1$ \textit{local transactions} $t_i$,
an element in block $v_i^j$.
The granularity of blockchain growth (and suspension when forks occur) is a block,
and if we assume our interest is in a single global transaction, we can represent the transaction with the set of involved blocks:
$T = \bigcup_{0 \leq i \leq n} v_i^j$.

We use $\mathtt{dim}$ and $\mathtt{skel}^k$ as the function operators of a simplex's dimension and $k$-skeleton, respectively.
We use $|\sigma|$ to denote the geometric realization, i.e. the polygon, of (abstract) simplex $\sigma$.
The $N$-time Barycentric and Chromatic subdivisions are denoted $\mathtt{Bary}^N$ and $\mathtt{Ch}^N$, respectively.
A complete list of notations and definitions in combinatorial topology can be found in~\cite{mherlihy_book13}. 
We assume an asynchronous, message-passing communication model among blockchains.
We only consider crash failures in this preliminary study and assume the number of faulty nodes $t$ is less than 50\%:
$t < \frac{n + 1}{2}$ in blockchains.

\subsection{Task}

A task of CBT is represented by a triple ($\cal I$, $\cal{O}$, $\Delta$),
where $\cal I$ is the input simplical complex,
$\cal O$ is the output simplical complex,
and $\Delta$ is the carrier map 
$\Delta: \mathcal{I} \rightarrow 2^\mathcal{O}$.

Each vertex, i.e., 0-simplex, in $\cal I$ is a tuple in the form of $(v_i^j, val_{in})$,
where $v_i^j$, as defined in~\S\ref{subsec:cbt_notation}, is block-$j$ at blockchain-$i$ and $val_{in} \in \{0, 1, \bot\}$.
The meaning in the input set is as follows,
0: local transaction not committed,
1: local transaction committed, and
$\bot$: the branch where this block resides is suspended.
There is an edge, i.e., 1-simplex, between every pair of vertices in $\cal I$ except that both vertices are the same block.
In general, an $l$-simplex in $\cal I$ comprises a set of distinct $l+1$ blocks as vertices and the higher-dimensional $k$-skeletons, $1 \leq k \leq l$. 
Overall, for a $(n+1)$-blockchain transaction, the input complex $\cal I$ comprises $3(n+1)$ vertices and simplices of dimension up to $n$, i.e., $\mathtt{dim}(\mathcal{I})$ = $n$.

Each vertex in $\cal O$ is a tuple $(v_i^j, val_{out})$,
where $v_i^j$ is, again, a specific block and $val_{out} \in \{1, 0\}$ with the same semantics defined for $val_{in}$.
Indeed, all of local transactions in $T$ should only end up with either \textit{committed} (1) or \textit{aborted} (0),
respecting the atomicity requirement.
The 1-simplices of $\cal O$ are all the edges connecting vertices whose $val_{out}$'s are equal, either 0 or 1,
among all blocks.
Therefore, by definition, the output simplicial complex is disconnected and has two path-connected components:
the global transaction is either (i) successfully committed, or (ii) aborted without partial changes.

We now construct the carrier map $\Delta$,
which maps each simplex from $\cal I$ to a subcomplex of $\cal O$.
Without loss of generality, pick any $l$-simplex $\sigma \in \cal I$,
$0 \leq l \leq n$,
and $\Delta$ specifies:
\begin{itemize}
    \item If all the $l$ $val_{in}$'s in $\sigma$ are 1, then $\mathtt{skel}^0$ $\Delta(\sigma) = \{(v, 1): v \in $ $\mathtt{skel}^0 \sigma \}$.

    \item If any of the $l$ $val_{in}$'s in $\sigma$ is $\bot$, $\mathtt{skel}^0$ $\Delta(\sigma) = \{(v, 0): v \in $ $\mathtt{skel}^0 \sigma \}$.
    
    \item For other cases, $\mathtt{skel}^0$ $\Delta(\sigma) = \{(v, 0), (v, 1): v \in $ $\mathtt{skel}^0 \sigma \}$.
    
    \item Any $k$-face $\tau \in \sigma$, $0 \leq k \leq l$, is similarly mapped.
\end{itemize}

Note that, by definition, $\Delta$ is rigid: 
In any of the above three cases, for any $l$-simplex $\sigma \in \cal I$, $\mathtt{dim}(\Delta(\sigma)) = l$.
Evidently, $\Delta$ is monotonic: 
adding new simplices into $\sigma$ can only enlarge the mapped subcomplex in $\mathcal{O}$.
Furthermore, $\Delta$ is name-preserving as constructed. 
Therefore, $\Delta$ is a well-defined carrier map from $\cal I$ to $2^{\cal O}$.

\subsection{Solvability}

\begin{definition} [Colorless CBT]
A colorless version of CBT, $(\mathcal{I}, \mathcal{O}', \Xi)$, is defined similarly as the general, ``colored'' CBT, $(\mathcal{I}, \mathcal{O}, \Delta)$, without the block identities on vertices in $\cal O'$.
Also, no identity match is required for the carrier map $\Xi: \mathcal{I} \rightarrow 2^\mathcal{O'}$.
\end{definition}

\begin{lemma}
\label{lemma:no_conti_map}
For colorless CBT $(\mathcal{I}, \mathcal{O'}, \Xi)$, there does not exist a continuous map 
\\$f: |\mathtt{skel}^t \mathcal{I}| \rightarrow |\mathcal{O'}|$ carried by $\Xi$, where $0 < t < \frac{n + 1}{2}$.
\end{lemma}
\begin{proof}[Proof Sketch]
The condition $t < \frac{n + 1}{2}$ is trivially satisfied by the assumption of crash failures,
as the blockchains would have been hard forked otherwise.
Since we assume at least one fork suspension would occur, we have $t > 0$.
The input simplicial complex $\cal I$ is pure of dimension $n$ by construction,
meaning that $\mathtt{skel}^t \mathcal{I}$ is $(t-1)$-connected. 
Because $t > 0$, $\mathtt{skel}^t \mathcal{I}$ is at least 0-connected (i.e., path-connected).
As a result, the geometric realization $|\mathtt{skel}^t \mathcal{I}|$ must be connected. 
However, we know that $\cal O'$ has two disjoint connected components;
so $|\cal O'|$ is not connected.
Therefore, a continuous map carried by $\Xi$ does not exist.
\end{proof}

\begin{lemma}
\label{lemma:colorless_no_protocol}
Colorless CBT $(\mathcal{I}, \mathcal{O'}, \Xi)$ does not have a t-resilient message-passing protocol.
\end{lemma}
\begin{proof}[Proof Sketch]
For contradiction, suppose a protocol solves task $(\mathcal{I}, \mathcal{O'}, \Xi)$.
Then we know that, after $N$ times of Barycentric subdivisions, the carrier map can be written in this form $\Xi(\sigma) = \mathtt{Bary}^N \mathtt{skel}^t \sigma$, for $\sigma \in \mathcal{I}$.
That is, there exists a carrier map $\Phi: \mathtt{Bary}^N \mathtt{skel}^t \mathcal{I} \rightarrow 2^\mathcal{O'}$.
Taking the geometric realizations, we thus have a continuous map $f = |\Phi|: |\mathtt{Bary}^N \mathtt{skel}^t \mathcal{I}| \rightarrow |\mathcal{O'}|$.
Note that a subdivision does not change the geometric realization: 
$|\mathtt{Bary}^N \mathtt{skel}^t \mathcal{I}| = |\mathtt{skel}^t \mathcal{I}|$.
Thus, we have $f: |\mathtt{skel}^t \mathcal{I}| \rightarrow |\mathcal{O'}|$,
a contradiction to Lemma~\ref{lemma:no_conti_map}.
\end{proof}

\begin{lemma} 
\label{lemma:reduction}
A model for colorless CBT $(\mathcal{I}, \mathcal{O}', \Xi)$ reduces to one for general CBT $(\mathcal{I}, \mathcal{O}, \Delta)$.
\end{lemma}
\begin{proof}[Proof Sketch]
Suppose a protocol $P$ solves $(\mathcal{I}, \mathcal{O}, \Delta)$,
we simulate $P$ with a protocol $P'$ for $(\mathcal{I}, \mathcal{O}', \Xi)$ as follows.
For any $l$-simplex in $\cal O$,
we drop the prefix of the $l$ vertices with map $\varphi: \mathbb{Z} \times V \rightarrow V$ such that
$(k, val_{out}) \mapsto (val_{out}) \in \mathcal{O'}$,
$0 \leq k \leq l$.
The carrier map in the colorless counterpart is
$\Xi = \Delta \circ \varphi$,
such that for $\sigma \in \mathcal{I}$, 
$\Xi(\sigma) = \Delta(\varphi(\sigma)) \subseteq \Delta(\sigma)$,
i.e., $\Xi$ is carried by $\Delta$.
\end{proof}

\begin{proposition}
For $t < \frac{n + 1}{2}$, $(\mathcal{I}, \mathcal{O}, \Delta)$ does not have a t-resilient message-passing protocol.
\end{proposition}
\begin{proof}
The claim follows directly from Lemma~\ref{lemma:colorless_no_protocol} and Lemma~\ref{lemma:reduction}.
\end{proof}



\bibliography{ref}

\end{document}